\documentclass[10pt, a4paper]{article}

\usepackage{microtype}
\usepackage{graphicx}
\usepackage{subfigure}
\usepackage{booktabs} 

\usepackage{algorithm}
\usepackage{algpseudocode}

\usepackage{amsmath}
\usepackage{amssymb}
\usepackage{amsthm}

\usepackage{graphicx}
\usepackage{subfigure}
\usepackage{epstopdf}
\usepackage{wrapfig}
\usepackage{tikz,pgfplots}
\usepackage{pgfplotstable}
\pgfplotsset{compat=1.8}
\usetikzlibrary{pgfplots.statistics}

\usepackage{verbatim}

\usepackage{url}

\newcommand{\prefu}{\stackrel{u{\succ}}}

\newcommand{\SO}{\texttt{SAROS}}
\newcommand{\BPR}{\texttt{BPR}}
\newcommand{\MostPop}{\texttt{MostPop}}

\newcommand{\MF}{\texttt{MF}}
\newcommand{\GRU}{\texttt{GRU4Rec}}
\newcommand{\ProdVec}{\texttt{Prod2Vec}}

\newcommand{\batch}{\texttt{BPR$_b$}}
\newcommand{\caser}{\texttt{Caser}}

\newcommand{\SASR}{\texttt{SASRec}}

\newcommand{\mapk}{\texttt{MAP@K}}
\newcommand{\apk}{\texttt{AP@K}}
\newcommand{\mapfive}{\texttt{MAP@5}}
\newcommand{\mapten}{\texttt{MAP@10}}
\newcommand{\ndcgfive}{\texttt{NDCG@5}}
\newcommand{\ndcgten}{\texttt{NDCG@10}}
\newcommand{\ndcgk}{\texttt{NDCG@K}}
\newcommand{\dcgk}{\texttt{DCG@K}}
\newcommand{\idcgk}{\texttt{IDCG@K}}

\newcommand{\NetF}{\textsc{Netflix}}
\newcommand{\RecS}{\textsc{RecSys'16}}
\newcommand{\Out}{\textsc{Outbrain}}
\newcommand{\ML}{\textsc{ML}}
\newcommand{\kasandr}{\textsc{Kasandr}}
\newcommand{\PANDOR}{\textsc{Pandor}}

\newcommand{\calI}{\mathcal{I}} 
\newcommand{\calS}{\mathcal{U}} 
\newcommand{\calB}{\mathcal{B}} 

\newcommand{\bR}{\mathbb{R}} 

\newcommand{\bfU}{\bar U}
\newcommand{\bfI}{\bar I} 

\newcommand{\calL}{\mathcal{L}} 
\newcommand{\weight}{\omega}

\newcommand{\posI}{\mathcal{I}^{+}} 
\newcommand{\negI}{\mathcal{I}^{-}} 

\newtheorem{theorem}{Theorem}

\newtheorem{assumption}{Assumption}
\usepackage[T1]{fontenc}
\usepackage[utf8]{inputenc}
\usepackage[english]{babel}
\usepackage{authblk}
\usepackage{chicago}

\title{Learning over no-Preferred and Preferred Sequence of items for Robust Recommendation}
\author[1,6]{Aleksandra Burashnikova \thanks{Corresponding author, aleksandra.burashnikova@skoltech.ru}}
\author[2]{Marianne Clausel}
\author[3]{Charlotte Laclau}
\author[4]{Frack Iutzeller}
\author[1,5]{Yury Maximov}
\author[6]{Massih-Reza Amini}
\affil[1]{Skolkovo Institute of Science and Technology}
\affil[2]{Department of Statistics University of Lorraine}
\affil[3]{Hubert Curien Laboratory Telecom Saint-Etienne}
\affil[4]{Jean Kuntzmann Laboratory, University Grenoble Alpes}
\affil[5]{Theoretical Division T-5, Los Alamos National Laboratory}
\affil[6]{Laboratoire d'Informatique de Grenoble, University Grenoble-Alpes}

\date{\today}

\begin{document}

\maketitle

\begin{abstract}
In this paper, we propose a theoretically founded sequential strategy for training large-scale Recommender Systems (RS) over implicit feedback, mainly in the form of clicks. The proposed approach consists in minimizing pairwise ranking loss over blocks of consecutive items constituted by a sequence of non-clicked items followed by a clicked one for each user. We present two variants of this strategy where model parameters are updated using either the momentum method or a gradient-based approach. To prevent from updating the parameters for an abnormally high number of clicks over some targeted items (mainly due to bots), we introduce an upper and a lower threshold on the number of updates for each user. These thresholds are estimated over the distribution of the number of blocks in the training set. The thresholds affect the decision of RS and imply a shift over the distribution of items that are shown to the users. Furthermore, we provide a convergence analysis of both algorithms and demonstrate their practical efficiency over six large-scale collections, both regarding different ranking measures and computational time.
\end{abstract}

\section{Introduction}
\label{Introduction}

With the increasing number of products available online, there is a surge of interest in the design of automatic systems --- generally referred to as Recommender Systems (RS) --- that provide personalized recommendations to users by adapting to their taste. The study of RS has become an active area of research these past years, especially since the  Netflix Price \cite{Bennett07}. One characteristic of online recommendation is the huge unbalance between the available number of products and those shown to the users. On the other hand, bots that interact with the system by providing too much feedback over some targeted items \cite{7824865}. Contrariwise, many users do not interact with the system over the items that are shown to them. In this context, the main challenges concern the design of a scalable and an efficient online RS in the presence of noise and unbalanced data. These challenges have evolved in time with the continuous development of data collections released for competitions or issued from e-commerce\footnote{
\url{https://www.kaggle.com/c/outbrain-click-prediction}}.  Recent approaches for RS \cite{wang2020setrank,8616805} now primarily consider feedback, mostly in the form of clicks that are easier to collect than {\it explicit} feedback, which is in the form of scores. Implicit feedback is more challenging to deal with as they do not depict the preference of a user over items, i.e., (no)click does not necessarily mean (dis)like \cite{Hu:2008}. In this case, most of the developed approaches are based on the Learning-to-rank paradigm and focus on how to leverage the click information over the unclick one without considering the sequence of users' interactions. 

In this paper, we propose $\SO_{}$, a sequential strategy for recommender systems with implicit feedback that updates  model parameters user per user over blocks of items constituted by a sequence of unclicked items followed by a clicked one. We present two variants of this strategy.  The first one, referred to as $\SO_m$, updates  model parameters at each time a block of unclicked items followed by a clicked one is formed after a user's interaction. Parameters' updates are carried out by minimizing the average ranking loss of the current model that scores the clicked item below the unclicked ones  using a momentum method \cite{PolyakB,nemirovskii1985optimal}. The second strategy, referred to as $\SO_b$, updates model parameters by minimizing the ranking loss over the same blocks of unclicked items followed by a clicked one using a gradient descent approach; with the difference that parameter updates are discarded for users who interact very little or a lot with the system.

In this paper, 
\begin{itemize}
    \item We propose a unified framework in which we study the convergence properties of both  versions of \SO{} in the general case of non-convex ranking losses. This is an extension of the work of \cite{Burashnikova19}, where only the convergence of $\SO_b$ was studied in the case of convex ranking losses.
    \item Furthermore, we provide empirical evaluation over six large publicly available datasets showing that both versions of $\SO_{}$ are highly competitive compared to the state-of-the-art models in terms of quality metrics and, that are significantly faster than both the batch and the online versions of the algorithm.
\end{itemize}

The rest of this paper is organized as follows. Section \ref{sec:soa} relates our work to previously proposed approaches. Section \ref{sec:Frame} introduces the general ranking learning problem that we address in this study. Then, in Section~\ref{sec:TA}, we present both versions of the $\SO_{}$ algorithm, $\SO_b$ and $\SO_m$, and provide an analysis of their convergence. Section \ref{sec:Exps} presents experimental results that support our approach. Finally, in Section \ref{sec:Conclusion}, we discuss the outcomes of this study and give some pointers to further research.

\section{Related work}
\label{sec:soa}

Two main approaches have been proposed for recommender systems. The first one, Content-Based recommendation or cognitive filtering \cite{Pazzani2007}, makes use of existing contextual information about the users (e.g., demographic information) or items (e.g., textual description) for the recommendation. The second approach, Collaborative Filtering, is undoubtedly the most popular one~\cite{Su:2009}, relies on past interactions and recommends items to users based on the feedback provided by other similar users. 

Traditionally, collaborative filtering systems have been designed using  {\it explicit} feedback, mostly in the form of rating~\cite{Koren08}. However,  rating information is non-existent on most of e-commerce websites and is challenging to collect, and user interactions are often done sequentially. Recent RS systems focus on learning scoring functions using {\it implicit} feedback to assign higher scores to clicked items than to unclicked ones rather than to predict the clicks as it is usually the case when we are dealing with explicit feedback~\cite{He2016,rendle_09,Zhang:16}. The idea here is that even a clicked item does not necessarily express the preference of a user for that item, it has much more value than a set of unclicked items for which no action has been made.

In most of these approaches, the objective is to rank the clicked item higher than the unclicked ones by finding a suitable representation of users and items in a way that for each user the ordering of the clicked items over unclicked ones is respected by dot product in the joint learned space. One common characteristic of publicly available collections for recommendation systems is the huge unbalance between positive (click) and negative feedback (no-click) in the set of items displayed to the users, making the design of an efficient online RS extremely challenging. Some works propose to reweight the impact of positive and negative feedback directly in the objective function \cite{Pan:2008} to improve the quality. Another approach is to sample data over a predefined set of interactions before learning \cite{Liu2016}. 

 Many new approaches tackle the sequential learning problem for RS by taking into account the temporal aspect of interactions directly in the design of a dedicated model and are mainly based on Markov Models (MM), Reinforcement Learning (RL), and Recurrent Neural Networks (RNN) \cite{Donkers:2017}. Recommender systems based on Markov Models, consider a subsequent interaction of users as a stochastic process over discrete random variables related to predefined user behavior. These approaches suffer from some limitations, mainly due to the sparsity of the data leading to a poor estimation of the transition matrix \cite{GuyShani} and choice of an appropriate order for the model \cite{He}.
Various strategies have been proposed to leverage the limitations of Markov Models. For instance, \cite{GuyShani} proposes to consider only the last frequent sequences of items and using finite mixture models. \cite{He} suggests combining similarity-based methods with high-order Markov Chains. Although it has been shown that in some cases, the proposed approaches can capture the temporal aspect of user interactions, these models suffer from a high time-complexity and do not pass the scale.  Some other methods consider RS as a Markov decision process (MDP) problem and solve it using reinforcement learning (RL) \cite{Moling,Tavakol}.
The size of discrete actions bringing the RL solver to a larger class of problems is also a bottleneck for these approaches. Recently many Recurrent Neural Networks (RNN) such as GRU or LSTM have been proposed for personalized recommendations \cite{hidasi2018recurrent,tang2018caser,DBLP:conf/icdm/Kang18}.
In this approach, the input of the network is generally the sequence of user interactions consisted of a single behaviour type (click,  adding to favourites, purchase, etc.) and the output is the predicted preference over items in the form of posterior probabilities of the considered behaviour type given the items. A comprehensive survey of Neural Networks based sequential approaches for personalized recommendation is presented in \cite{Fang:20}. All these approaches do not consider negative interactions; i.e. viewed items that are not clicked or purchased; and the system's performance on new test data may be affected.

Our approach differs from other sequential based methods in the way that the model parameters are updated, at each time a block of unclicked items followed by a clicked one is constituted. This update scheme follows the hypothesis that user preference is not absolute over the items which were clicked, but it is relative with respect to  items that were viewed. Especially, we suppose that a user may not prefer an item in absolute but may click on it if the item is shown in a given context respectively to other items that were shown but not have been clicked. For this update scheme we further propose two variants by $(1)$ either smoothing the parameter updates with the momentum technique; or $(2)$ controlling the number of blocks per user interaction. 
That is, if for a given user the number of blocks is below or above two predefined thresholds found over the distribution of the number of blocks,  parameter updates for the user are discarded. We further provide a proof of convergence of both variants of the proposed approach in the general case of non-convex loss functions.

\section{Framework and Problem Setting}
\label{sec:Frame}

Throughout, we use the following notation. For any positive integer $n$, $[n]$ denotes the set $[n]= \{1,\ldots,n\}$. We suppose that $\calI= [M]$ and $\calS= [N]$ are two sets of indexes defined over items and users. Further, we assume that each pair constituted by a user $u$ and an item $i$ is identically and independently distributed (i.i.d) according to a fixed yet unknown distribution ${\cal D}$. 
At the end of his or her session, a user $u\in\calS$ has reviewed a subset of items $\calI_u\subseteq \calI$ that can be decomposed into two sets: the set of preferred and non-preferred items denoted by $\posI_u$ and $\negI_u$, respectively. Hence, for each pair of items $(i,i')\in\posI_u\times \negI_u$, the user $u$ prefers item $i$ over item $i'$; symbolized by the relation $i\!\prefu\! i'$. From this preference relation a desired output $y_{u,i,i'}\in\{-1,+1\}$ is defined over the pairs $(u,i)\in\calS\times\calI$ and $(u,i')\in\calS\times\calI$, such that $y_{u,i,i'}=+1$ if and only if   $i\!\prefu\! i'$.  We suppose that the indexes of users as well as those of items in the set $\calI_u$, shown to the active user $u\in\calS$, are ordered by time. 

Finally, for each user $u$, parameter updates are performed over blocks of consecutive items where a  block $\calB_u^t=\text{N}_u^{t}\sqcup\Pi_u^{t}$, corresponds to a time-ordered sequence (w.r.t. the time when the interaction is done) of no-preferred items, $\text{N}_u^{t}$, and at least one preferred one, $\Pi_u^{t}$. Hence, $\posI_u=\bigcup_t \Pi_u^{t}$ and $\negI_u=\bigcup_t \text{N}_u^{t}; \forall u\in\calS$. Notation is summarized in Table~$\ref{table:notation}$.

\subsection{Learning Objective}
\label{sec:LO}
Our objective here is to minimize an expected error penalizing the misordering of all pairs of interacted items $i$ and $i'$ for a user $u$. Commonly, this objective is given under the Empirical Risk Minimization (ERM) principle, by minimizing the empirical ranking loss estimated over the items and the final set of users who interacted with the system~:
\begin{align}\label{eq:RL}
    \hat{\calL}_u(\weight)\!=\!\frac{1}{|\posI_u||\negI_u|}\!\sum_{i\in \posI_u}\!\sum_{i'\in \negI_u} \!\ell_{u,i,i'} (\weight),
\end{align}
\noindent where, $\ell_{u, i, i'} (.)$ is an instantaneous ranking loss defined over the triplet $(u,i,i')$ with    
 $i\!\prefu\! i'$. Hence,  $\hat{\calL}_u(\weight)$ is the pairwise ranking loss with respect to user's interactions and $\calL(\weight)~=~{\mathbb E}_{u} \left[\hat{\calL}_u(\weight)\right]$ is the expected ranking loss, where ${\mathbb E}_{u}$ is the expectation with respect to users chosen randomly according to the marginal distribution. 
 \begin{table}[t!]
\centering
 \begin{tabular}{c|l} 
 $\calI= [M]$ & The set of item indexes\\ \hline
 $\calS= [N]$ & The set of user indexes\\ \hline
 ${\cal D}$ & joint distribution over users and items \\ \hline
 ${\cal D}_{u}$ & conditional distribution of items  for a fixed user $u$\\ \hline
 $\text{N}_u^{t}$ & Negative items in block $t$ for user $u$\\\hline
 $\Pi_u^{t}$ & Positive items in block $t$ for user $u$\\\hline
 $\calB_u^t=\text{N}_u^{t}\sqcup\Pi_u^{t}$ & Negative and positive items in block $t$ for user $u$\\\hline
 $\posI_u$ & The set of all positive items for user $u$\\ \hline
 $\negI_u$ & The set of all negative items for user $u$\\ \hline 
 $\ell_{u, i, i'}(\omega)$ & Instantaneous loss for user $u$ and a pair of items $(i, i')$\\ \hline 
 $\hat{\calL}_u(\omega)$ & Empirical ranking loss with respect to user $u$\\
 & \qquad $\hat{\calL}_u(\omega) = \frac{1}{|\posI_u||\negI_u|}\!\sum_{i\in \posI_u}\!\sum_{i'\in \negI_u} \!\ell_{u,i,i'} (\omega)$\\
 \hline 
 ${\hat {\calL}}_{\calB^t_u}(\omega)$ &  Empirical ranking loss with respect to a block of items \\ 
 & \qquad ${\hat {\calL}}_{\calB^t_u}(\omega) = \frac{1}{|\Pi_u^{t}||\text{N}_u^{t}|}\sum_{i \in \Pi_u^{t}} \sum_{i'\in \text{N}_u^{t}} \ell_{u, i, i'} (\omega)$ \\
 \hline
 ${\calL}(\omega)$ & Expected ranking loss  ${\calL}(\omega) = \mathbb{E}_{{\cal D}_u}\hat{\calL}_u(\omega)$
 \end{tabular}
 \caption{Notation}
 \label{table:notation}
\end{table}
As in other studies, we represent each user $u$ and each item $i$ respectively by vectors $\bfU_u\in\bR^k$ and $\bfI_i\in\bR^k$ in the same latent space of dimension $k$ \cite{Koren:2009}. The set of weights to be found $\weight=(\bfU,\bfI)$, are then matrices formed by the vector representations  of users $\bfU=(\bfU_u)_{u\in [N]}\in\bR^{N\times k}$ and items $\bfI=(\bfI_i)_{i\in[M]}\in\bR^{M\times k}$. 
The minimization of the ranking loss above in the batch mode with the goal of finding user and item embeddings, such that the dot product between these representations in the latent space reflects the best the preference of users over items, is a common approach. Other strategies have been proposed for the minimization of the empirical loss \eqref{eq:RL}, among which the most popular one is perhaps the Bayesian Personalized Ranking (\BPR) model \cite{rendle_09}. In this approach, the instantaneous loss, $\ell_{u,i,i'}$, is the surrogate regularized logistic  loss for some hyperparameter $\mu \ge 0$:
\begin{align}
\label{eq:instloss}
    \ell_{u,i,i'}(\weight) =  \log\left(1+e^{-y_{u,i,i'}\bfU_u^\top(\bfI_{i}-\bfI_{i'})}\right)+\lambda (\|\bfU_u\|_2^2+\|\bfI_{i}\|_2^2 + \|\bfI_{i'}\|_2^2)
\end{align}

The {\BPR} algorithm proceeds by first randomly choosing a user $u$, and then repeatedly selecting two pairs $(i,i')\in \calI_u\times \calI_u$.
In the case where one of the chosen items is preferred over the other one (i.e., $y_{u,i,i'}\in\{-1,+1\}$), the algorithm then updates the weights using the stochastic gradient descent method over the instantaneous loss \eqref{eq:instloss}.  

\subsection{Algorithm \SO{}}
\label{sec:Algo}

Another aspect is that user preference over items depend mostly on the context where these items are shown to the user. A user may prefer (or not) two items independently one from another, but within a given set of shown items, he or she may completely have a different preference over these items. By  randomly sampling triplets constituted by a user and corresponding clicked and unclicked items selected over the whole set of shown items to the user, this effect of local preference is not taken into account. Furthermore, triplets corresponding to different users are non uniformly distributed, as interactions vary from one user to another user, and for parameter updates; triplets corresponding to low interactions have a small chance to be chosen. In order to tackle these points; we propose to update the parameters sequentially over the blocks of non-preferred items followed by preferred ones for each user $u$. The constitution of $B+1$ sequences of non-preferred and preferred blocks of items for two users $u$ and $u+1$ is shown in Figure \ref{fig:SILICOM}. 

\begin{figure}[!th]
    \centering
    \includegraphics[width=.95\textwidth]{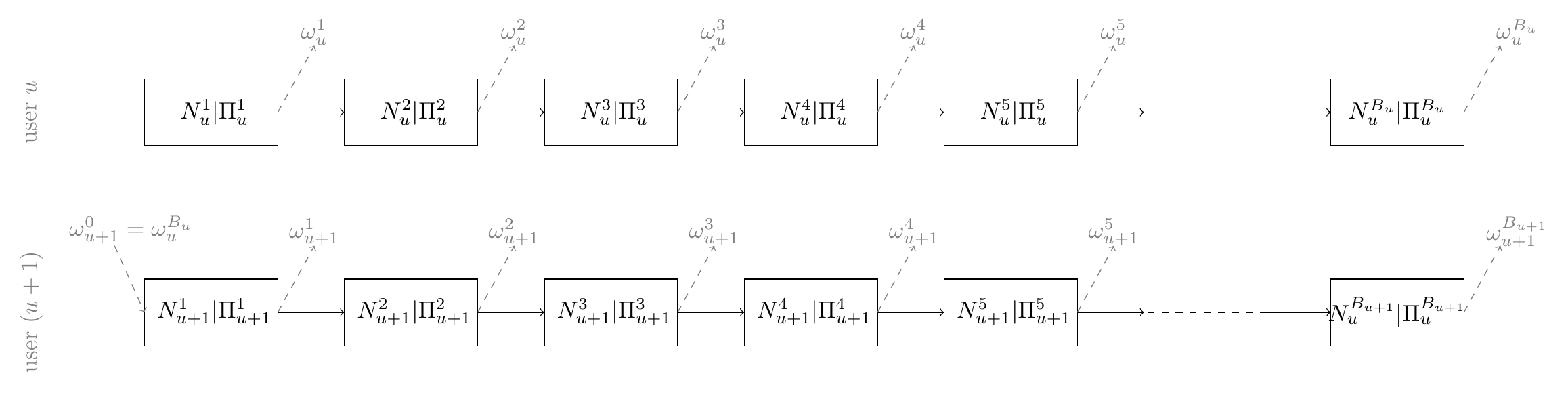}
    \caption{The horizontal axis represents the sequence of interactions over items ordered by time. Each update of weights $\omega_u^{t}; t\in\{b,\ldots,B\}$ occurs whenever the corresponding sets of negative interactions, $\text{N}^t_u$, and positive ones, $\Pi_u^t$, exist.}
    \label{fig:SILICOM}
\end{figure}

In this case, at each time  a block $\calB_u^t=\text{N}_u^{t}\sqcup\Pi_u^{t}$ is formed for user $u$; weights are updated by miniminzing the ranking loss corresponding to this block~:

\begin{equation}
\label{eq:CLoss}
    {\hat {\calL}}_{\calB_u^t}(\weight_u^{t}) = \frac{1}{|\Pi_u^{t}||\text{N}_u^{t}|}\sum_{i \in \Pi_u^{t}} \sum_{i'\in \text{N}_u^{t}} \ell_{u, i, i'} ({\weight}_u^{t}).
\end{equation}

We propose two strategies for the update of weights. In the first one, referred to as {\SO$_m$}, the aim is to carry out an effective minimization of the ranking loss \eqref{eq:CLoss} by lessening the oscillations of the updates through the minimum. This is done by defining the updates as the linear combination of the gradient of the loss \eqref{eq:CLoss}, $\nabla  \widehat{\calL}_{{\mathcal B}^t_u}(w_u^{t})$, and the previous update as in the momentum technique at each iteration $t$~:

\begin{align}
\label{thm11:eq1}
    v_u^{t+1}&=\mu\cdot v_u^{t}+(1-\mu)\nabla  \widehat{\calL}_{{\mathcal B}^t_u}(w_u^{t})\\
    w_u^{t+1}&= w_u^{t}-\alpha v_u^{t+1}
\end{align}

\noindent 
where $\alpha$ and $\mu$ are hyperparameters of the linear combination. In order to explicitly take into account bot attacks -- in the form of excessive clicks over some target items -- we propose a second variant of this strategy, referred to as {$\SO_b$}. This variant consists in fixing two thresholds $b$ and $B$  over the parameter updates. For a new user $u$, model parameters are updated if and only if the number of the  blocks of items constituted for this user is within the interval $[b,B]$. 

\begin{algorithm}[t!]
   \caption{$\SO_b$: SequentiAl RecOmmender System for implicit feedback}
   \label{alg:CC-conv}
\begin{algorithmic}
\State {\bfseries Input:} A sequence (user and items)  $\{(u,(i_1, \dots, i_{|\mathcal{I}_u|})\}_{u=1}^N$ drawn i.i.d. from~${\cal D}$; \\
    \hspace{15mm}A maximal $B$ and a minimal $b$ number of blocks allowed per user $u$;\\
    \hspace{15mm}A maximal number of epochs $E$;
\State {\bfseries Initialization:} Initialize parameters $\omega_1^0$ randomly;
\For{$e = 1 ... E$} \Comment{\color{gray}Loop over all epochs\color{black}}
    \For{$u = 1 ... N$} \Comment{\color{gray}Loop over all users $u$\color{black}}
        \State {$t\leftarrow 0$; $\calI_u=(i_1, \dots, i_{|\mathcal{I}_u|})$ the sequence of items viewed by the user $u$;}
        \State {$j\leftarrow 1; \text{N}_u^{t} \leftarrow \varnothing; \, \Pi_u^t \leftarrow \varnothing;$}
        \While{$t \leq B$ and $j\leq |\mathcal{I}_u|$} \Comment{\color{gray} Loop over items  displayed to user $u$\color{black}}
            \While{$feedback(u,i_j)=-1$ and $j\leq |\mathcal{I}_u|$}
                \Comment{\color{gray}{\scriptsize While $u$ has a negative feedback on $i_j$}\color{black}}
                \State $\text{N}_u^t \leftarrow \text{N}_u^t \cup \{i\}; j\leftarrow j+1$
            \EndWhile
            \While{$feedback(u,i_j)=+1$ and $j\leq |\mathcal{I}_u|$}
                \Comment{\color{gray}{\scriptsize While $u$ has a positive feedback on $i_j$}\color{black}}
                \If{ $\text{N}_u^{t}\neq \varnothing$}
                    \Comment{\color{gray}{\scriptsize If there were negative feedback before the positive ones}\color{black}}
                    \State $\Pi_u^t \leftarrow \Pi_u^t \cup \{i_j\}$; 
                \EndIf
                \State $j\leftarrow j+1;$
            \EndWhile
            \If{ $\text{N}_u^{t} \neq \varnothing$ and $\Pi_u^{t} \neq \varnothing$}
            \Comment{\color{gray}{\scriptsize If  negative and positive blocks are constituted}\color{black}}
                \State {\small $\weight_u^{t+1} \leftarrow\weight_u^t - \frac{\eta}{|{\text{N}_u^{t}}||\Pi_u^{t}|}  \displaystyle{\sum_{i\in\Pi_u^{t}}\sum_{i'\in \text{N}_u^{t}}} \nabla \ell_{u,i,i'} (\weight_{u}^t)$;}
                \State $t \leftarrow t+1; \text{N}_u^t \leftarrow \varnothing; \Pi_u^t \leftarrow \varnothing;$ 
            \EndIf
        \EndWhile
        \If{$t \leq b$}
            \Comment{\color{gray}{\scriptsize If the number of blocks is less than $b$, do not consider the updates }\color{black}}
            \State $\weight_{u}^t \leftarrow \weight_{u}^0$;
        \EndIf
        \If{$u<N$}
            \Comment{\color{gray}{\scriptsize Initialize the weights for the next user with the current ones}\color{black}}
            \State $\weight_{u+1}^0 \leftarrow \weight_{u}^t$;
        \Else 
            \Comment{\color{gray}{\scriptsize The next user to the last one in the  list of the users, is the first user}\color{black}}
            \State $\weight_{1}^0 \leftarrow \weight_{N}^t$;
        \EndIf
    \EndFor
\EndFor
\State {\bfseries Return:} The last updated weights; 
\end{algorithmic}
\end{algorithm}

The pseudo-code of {$\SO_b$} is shown in the following. Starting from initial weights $\weight_1^0$ chosen randomly for the first user. The sequential update rule, for each current user $u$ consists in updating the weights by making one step towards the opposite direction of the gradient of the ranking loss estimated on the current block,  $\calB_u^t=\text{N}_u^{t}\sqcup\Pi_u^{t}$~:

\begin{equation}
\weight_u^{t+1} = \weight_u^t - \frac{\eta}{|{\text{N}_u^{t}}||\Pi_u^{t}|}  \displaystyle{\sum_{i\in\Pi_u^{t}}\sum_{i'\in \text{N}_u^{t}}} \nabla \ell_{u,i,i'} (\weight_{u}^t)
\end{equation}

For a given user $u$, parameter updates are discarded if their number (i.e. number of blocks $(\calB_u^t)_t$ for that user) falls outside the interval $[b,B]$. In this case, parameters are initialized with respect to the latest update before user $u$ and they are updated with respect to a new user's interactions.

\subsection{Convergence analysis}
\label{sec:TA}

We provide proofs of convergence for both variants of \SO{} under the typical hypothesis that the system is not instantaneously affected by the sequential learning of the weights. This hypothesis stipulates that the generation of items shown to users is independently and identically distributed with respect to some stationary in time underlying distribution~${\cal D}_{\cal I}$, which constitutes the main hypothesis of almost all the existing studies. More precisely, our work is based on the following hypothesis.  


\begin{assumption}\label{ass:1}
    For any $u, t \ge 1$, we have 
    \begin{enumerate}
        \item 
        $\mathbb{E}_{(u,{\cal B}_u^t)} \|\nabla {\calL}(\omega_u^t) - \nabla \hat{\calL}_{{\cal B}_u^t}(\omega_u^t)\|_2^2 \le \sigma^2$,
        \item  For any $u$,  $\left|\mathbb{E}_{ {{\cal B}_u^t}|u} \langle\nabla {\calL}(\omega_u^t), \nabla {\calL}(\omega_u^t) -  \nabla \hat{\calL}_{{\cal B}_u^t}(\omega_u^t) \rangle \right| \le a^2 \|\nabla {\calL}(\omega_u^t)\|_2^2$
    \end{enumerate}
    for some parameters $\sigma>0$ and $a\in [0,1/2)$ independent of $u$ and $t$. 
\end{assumption}

The first assumption is common in stochastic optimization and it implies consistency of the sample average approximation of the gradient. However, the first assumption is not sufficient to prove the convergence because of interdependency of different blocks of items for the same user. 

The second assumption implies that in the neighborhood of the optimal point we have $\nabla {\calL}(\omega_u^t)^\top \nabla \hat{\calL}_{{\cal B}_u^t}(\omega_u^t) \approx \|\nabla {\calL}(\omega_u^t)\|_2^2$, which greatly helps to establish consistency and convergence rates for both variants of the methods. In particular, if an empirical estimate of the loss over a block is unbiased, e.g. $\mathbb{E}_{\calB_u^t} \nabla {\hat \calL}_{\calB_u^t}(\omega) = \nabla \calL (\omega)$, the second assumption is satisfied with $a = 0$. 

The following theorem establishes the convergence rate for the \SO{}$_b$ algorithm. 

\begin{theorem}\label{thrm:new-01}
Let $\ell$ be a (possibly non-convex) $\beta$-smooth loss function. Assume, moreover,  that the number of interactions per user belongs to an interval $[b, B]$ almost surely and assumption~\ref{ass:1} is satisfied with some constants $\sigma^2$ and $a$, $0 < a < 1/2$. 
Then, for a step-size policy $\eta_u^t \equiv \eta_u$ with $\eta_u\leq 1/(B\beta)$ for any user $u$, one has
\begin{align}
\min_{u:\, 1\le u \le N}\mathbb{E}    \|\nabla {\calL}(\omega_u^0)\|_2^2 \le  \frac{2({\calL}(\omega_1^0) - {\calL}(\omega_u^0)) + \beta \sigma^2 \sum_{u=1}^N \sum_{t=1}^{|{\cal B}_u|}(\eta_u^t)^2}{\sum_{u=1}^N \sum_{t=1}^{|{\cal B}_u|} \eta_u^t(1 - a^2 - \beta \eta_u^t(1/2 - a^2))}.
\end{align}
In particular, for a constant step-size policy $\eta_u^t = \eta = c/\sqrt{N}$ satisfies $\eta \beta \le 1$, one has 
\begin{align*}
    \min_{t, u} \|\nabla {\calL}(\omega_{u}^t)\|_2^2 \le 
    \frac{2}{b(1-4a^2)} \frac{2 ({\calL}(\omega_1^0) - {\calL}(\omega_*))/c + \beta c \sigma^2 B}{\sqrt{N}}. 
\end{align*}
\end{theorem}

\begin{proof}
Since $\ell$ is a $\beta$ smooth function, we have for any $u$ and $t$:
\begin{align*}
    {\calL}(\omega_{u}^{t+1}) & \le {\calL}(\omega_{u}^t) + \langle\nabla {\calL}(\omega_{u}^t), \omega_u^{t+1} - \omega_u^{t} \rangle + \frac{\beta}{2}(\eta_u^t)^2 \|\nabla \hat{\calL}_{{\cal B}_u^t}(\omega_u^t)\|_2^2
    \\
    & = {\calL}(\omega_{u}^t) - \eta_u^t \langle \nabla {\calL}(\omega_{u}^t), \nabla \hat{\calL}_{{\cal B}_u^t}(\omega_u^t) \rangle +   \frac{\beta}{2}(\eta_u^t)^2 \|\nabla \hat{\calL}_{{\cal B}_u^t}(\omega_u^t)\|_2^2
\end{align*}    
Following~\cite{lan2020first}; by denoting $\delta_u^t = \nabla \hat{\calL}_{{\cal B}_u^t}(\omega_u^t) - \nabla {\calL}(\omega_u^t)$, we have: 
 \begin{align}\label{eq:01}
    {\calL}(\omega_{u}^{t+1})   & \le {\calL}(\omega_{u}^t) - \eta_u^t \langle \nabla {\calL}(\omega_{u}^t), \nabla {\calL}(\omega_{u}^t) + \delta_u^t \rangle +   \frac{\beta}{2}(\eta_u^t)^2 \|\nabla {\calL}(\omega_u^t) + \delta_u^t\|_2^2\nonumber\\
    & = 
    {\calL}(\omega_{u}^t) 
    +  \frac{\beta(\eta_u^t)^2}{2}\|\delta_u^t\|_2^2 
    - \!\!\left(\eta_u^i - \frac{\beta(\eta_u^t)^2}{2}\right)\!\!\|\nabla {\calL}(\omega_{u}^t)\|_2^2 \nonumber\\ 
    & \hspace{60mm} - \left(\eta_u^t - \beta (\eta_u^t)^2\right) \langle\nabla {\calL}(\omega_u^t), \delta_u^t\rangle 
\end{align}

\noindent Our next step is to take the expectation on  both sides of inequality~\eqref{eq:01}. According to  Assumption~\ref{ass:1}, one has for some $a\in [0, 1/2)$:
\begin{align*}
    \left(\eta_u^t - \beta (\eta_u^t)^2\right)\left|\mathbb{E} \langle\nabla {\calL}(\omega_u^t), \delta_u^t\rangle\right| \le \left(\eta_u^t - \beta (\eta_u^t)^2\right) a^2 \|\nabla \calL (\omega_u^t)\|_2^2, 
\end{align*}
where the expectation is taken over the set of blocks and users seen so far. 

Finally, taking the same expectation on both sides of inequality~\eqref{eq:01}, it comes:
\begin{align}\label{eq:02}
    {\calL}(\omega_{u}^{t+1}) &\le  {\calL}(\omega_{u}^{t}) + \frac{\beta}{2}(\eta_u^t)^2\mathbb{E}\|\delta_u^t\|_2^2 - 
    \eta_u^t(1 - \beta\eta_u^t/2  - a^2|1 - \beta\eta_u^t|) \|\nabla {\calL}(\omega_{u}^t)\|_2^2 \nonumber\\
    & \le  {\calL}(\omega_{u}^{t}) + \frac{\beta}{2}(\eta_u^t)^2\|\delta_u^t\|_2^2 - 
    \eta_u^t\underbrace{(1 - a^2 
    -    \beta \eta_u^t(1/2 - a^2))}_{:= z_u^t} \|\nabla {\calL}(\omega_{u}^t)\|_2^2 \nonumber\\
    & = {\calL}(\omega_{u}^{t}) + \frac{\beta}{2}(\eta_u^t)^2\|\delta_u^t\|_2^2 - 
    \eta_u^t z_u^t \|\nabla {\calL}(\omega_{u}^t)\|_2^2 \nonumber\\
    & = {\calL}(\omega_{u}^{t}) + \frac{\beta}{2}(\eta_u^t)^2\sigma^2 - 
    \eta_u^t z_u^t \|\nabla {\calL}(\omega_{u}^t)\|_2^2, 
\end{align}
where the second inequality is due to  $|\eta_u^t\beta|\leq 1$. Also, as $|\eta_u^t\beta|\leq 1$ and $a^2\in [0,1/2)$ one has $z_u^t>0$ for any $u,t$. Rearranging the terms, one has
\begin{align*}
    \sum_{u=1}^N\sum_{t=1}^{|{\cal B}_u|} \eta_u^t z_u^t \|\nabla {\calL}(\omega_{u}^t)\|_2^2 \le {\calL}(\omega_1^0) - {\calL}(\omega_*) + \frac{\beta \sigma^2}{2}  \sum_{u=1}^N\sum_{t=1}^{|{\cal B}_u|}(\eta_u^t)^2.
\end{align*}
and
\begin{align*}
    \min_{t, u} \|\nabla {\calL}(\omega_{u}^t)\|_2^2
    & 
    \le 
    \frac{{\calL}(\omega_1^0) - {\calL}(\omega_*) + \frac{\beta}{2}  \sum_{u=1}^N\sum_{t=1}^{|{\cal B}_u|}(\eta_u^t)^2 \sigma^2 }{\sum_{u=1}^N\sum_{t=1}^{|{\cal B}_u|} \eta_u^t z_u^t} \\
    & 
    \le \frac{{\calL}(\omega_1^0) - {\calL}(\omega_*) + \frac{\beta}{2}  \sum_{u=1}^N\sum_{t=1}^{|{\cal B}_u|}(\eta_u^t)^2 \sigma^2 }{\sum_{u=1}^N\sum_{t=1}^{|{\cal B}_u|} \eta_u^t (1 - a^2 - \beta \eta_u^t(1/2 - a^2))} 
\end{align*}
Where, $\omega_*$ is the optimal point. Then, using a constant step-size policy, $\eta_u^i = \eta$, and the bounds on a block size, $b\leq |{\cal B}_u|\leq B$, we get:
\begin{align*}
    \min_{t, u} \|\nabla {\calL}(\omega_{u}^t)\|_2^2 & \le \frac{{\calL}(\omega_1^0) - {\calL}(\omega_*) + \frac{\beta\sigma^2}{2}  N\sum_{u=1}^N\eta_u^2  }{b\sum_{u=1}^N\eta_u (1 - a^2 - \beta \eta_u (1/2 - a^2))} \\
    &
    \le  \frac{4{\calL}(\omega_1^0) - 4{\calL}(\omega_*) + 2\beta \sigma^2 B \sum_{u=1}^N\eta^2}{b(1 - 4a^2)\sum_{u=1}^N\eta} \\
    &
    \le 
    \frac{2}{b(1-4a^2)} \left\{\frac{2{\calL}(\omega_1^0) - 2{\calL}(\omega_*)}{N \eta} + \beta\sigma^2 B \eta\right\}. 
\end{align*}

Taking $\eta = c/\sqrt{N}$ so that $0 < \eta \le 1/\beta$, one has 
\begin{align*}
    \min_{t, u} \|\nabla {\calL}(\omega_{u}^t)\|_2^2 \le 
    \frac{2}{b(1-4a^2)} \frac{2 ({\calL}(\omega_1^0) - {\calL}(\omega_*))/c + \beta c \sigma^2 B}{\sqrt{N}}. 
\end{align*}
If $b = B = 1$, this rate matches up to a constant factor to the standard $O(1/\sqrt{N})$ rate of the stochastic gradient descent. 
\end{proof}
Note that the stochastic gradient descent strategy implemented in the Bayesian Personalized Ranking model (\BPR) \cite{rendle_09} also converges to the minimizer of the ranking loss ${\calL}(\omega)$ (Eq. \ref{eq:RL}) with the same rate.

The analysis of momentum algorithm \SO{}$_m$
is slightly more involved. We say that a function $f(x)$ satisfies the Polyak-\L{}ojsievich condition \cite{Polyak63,Lojas63,karimi2016linear} if the following inequality holds for some $\mu > 0$: 
\[
    \frac{1}{2} \|\nabla f(x)\|_2^2 \ge \mu (f(x) - f^*),
\]
where $f^*$ is a global minimum of $f(x)$.

\begin{theorem}\label{thrm:new-02}
Let ${\calL}(\omega)$ be a (possibly non-convex) $\beta$-smooth function which satisfies the Polyak-Lojasievich condition with a constant $\mu > 0$. 
Moreover, assume the number of interactions per user belongs to an interval $[b, B]$ almost surely for some positive $b$ and $B$, and Assumption \ref{ass:1} is satisfied with some $\sigma^2$ and $a$.
Then, for $N = \sum_{u = 1}^N |{\cal B}_u|$ 
and a constant step-size policy $\eta_u^t = \eta $ with $\eta \beta \le 1$, one has
\begin{align*}
    {\calL}(\omega_u^{t+1}) - {\calL}(\omega_*) 
    \le \exp(- \mu \eta N) ({\calL}(\omega_1^{0}) - {\calL}(\omega_*)) + \frac{\beta\sigma^2 \eta^2}{2(1 - \mu/\beta)}, \quad \eta\beta \le 1.
\end{align*}
where the estimation is uniform for all $a$, $0 \le a < 1/2$. 
 
In particular, if $\eta = c/\sqrt{N}$, under the same conditions one has
\begin{align*}
     {\calL}(\omega_u^t) - {\calL}(\omega_*) \le \exp(-\mu c \sqrt{N}) ({\calL}(\omega_1^0) - {\calL}(\omega_*)) + \frac{\beta\sigma^2c^2}{2(1 - \mu/\beta)N}. 
 \end{align*}
\end{theorem}

\begin{proof}
Similarly to the Theorem~\ref{thrm:new-01}, From Ineq.~\eqref{eq:02} we have: 
\begin{align*}
    {\calL}(\omega_u^{t+1}) \le {\calL}(\omega_u^{t}) + \frac{\beta}{2} (\eta_u^t)^2 \sigma^2 - \eta_u^t z_u^t \|\nabla {\calL}(\omega_u^{t})\|_2^2
\end{align*}
for $z_u^t = 1 - a^2 - \beta \eta_u^t(1/2 - a^2)>0$. Further, using the Polyak-Lojasievich condition, it comes: 
\begin{align*}
    - \eta_u^t z_u^t \|\nabla {\calL}(\omega_u^{t})\|_2^2 \le - 2\mu \eta_u^t z_u^t (\calL(\omega_u^t) - \calL(\omega_*)), 
\end{align*}
and
\begin{align*}
    {\calL}(\omega_u^{t+1}) - {\calL}(\omega_*) & \le {\calL}(\omega_u^{t}) - {\calL}(\omega_*) + \frac{\beta}{2} (\eta_u^t)^2 \sigma^2 - 2\mu \eta_u^t z_u^t({\calL}(\omega_u^{t}) - {\calL}(\omega_*))\\
    & \le ({\calL}(\omega_u^{t}) - {\calL}(\omega_*)) (1 - 2\mu\eta_u^t z_u^t) + \frac{\beta}{2} (\eta_u^t)^2 \sigma^2\\
    & \le \prod_{u}\prod_{t} (1 - 2\mu\eta_u^t z_u^t) ({\calL}(\omega_1^{0}) - {\calL}(\omega_*)) \\
    & + \frac{\beta\sigma^2}{2}\sum_{v\le u} (\eta_v^t)^2\prod_{v}\prod_{t} (1 - 2\mu\eta_v^t z_v^t)
\end{align*}

Finally, for a constant step-size policy, $\eta_u^t = \eta$, one has $z_u^t = z = 1 - a^2 - \beta\eta(1/2-a^2)$ and
\begin{align*}
    {\calL}(\omega_u^{t+1}) - {\calL}(\omega_*) \le (1 - 2\mu \eta z)^N ({\calL}(\omega_1^{0}) - {\calL}(\omega_*)) + \frac{\beta\sigma^2 \eta^2}{2(1 - 2\mu \eta z)},
\end{align*}
where the last term is given by summing the geometric progression. As $\beta\eta \le 1$ and $a< 1/2$ one has $z \ge 1/2$. Thus
\begin{align*}
    {\calL}(\omega_u^{t+1}) - {\calL}(\omega_*) & \le (1 - \mu \eta )^N ({\calL}(\omega_1^{0}) - {\calL}(\omega_*)) + \frac{\beta\sigma^2 \eta^2}{2(1 - \mu/\beta)}\\
    & \le \exp(- \mu \eta N) ({\calL}(\omega_1^{0}) - {\calL}(\omega_*)) + \frac{\beta\sigma^2 \eta^2}{2(1 - \mu/\beta)}, \quad \eta\beta \le 1.
\end{align*}
Taking $\eta = c/\sqrt{N}$ for some positive $c$ guarantees a rate of convergence $O(1/N)$. With a different choice of the step-size policy, rates almost up to $O(1/N^2)$ are possible; however, these rates imply $O(1/N)$ convergence for the norm of the gradient which matches the standard efforts of stochastic gradient descent under the Polyak-Lojasievich condition \cite{karimi2016linear}. 
\end{proof}

\section{Experimental Setup and Results}
\label{sec:Exps}

In this section, we provide an empirical evaluation of our optimization strategy on some popular benchmarks proposed for evaluating RS. All subsequently discussed components were implemented in Python3 using the TensorFlow library\footnote{\url{https://www.tensorflow.org/}.} and computed on Skoltech CDISE HPC cluster ZHORES \cite{DBLP:journals/corr/abs-1902-07490}. 
We first proceed with a presentation of the general experimental set-up, including a description of the datasets and the baseline models.

\paragraph{Datasets. } We report results obtained on five publicly available data\-sets, for the task of personalized Top-N recommendation on the following collections~:
\begin{itemize}
\item \ML-1M \cite{Harper:2015:MDH:2866565.2827872} and {\NetF} consist of user-movie ratings, on a scale of one to five, collected from a movie recommendation service and the Netflix company. The latter was released to support the Netflix Prize competition. For both datasets, we consider ratings greater or equal to $4$ as positive feedback, and negative feedback otherwise.
\item  We extracted a subset out of the {\Out} dataset from of the Kaggle challenge\footnote{\url{https://www.kaggle.com/c/outbrain-click-prediction}} that consisted in the recommendation of news content to users based on the 1,597,426 implicit feedback collected from multiple publisher sites in the United States.
\vspace{1mm}\item  {\kasandr}\footnote{\url{https://archive.ics.uci.edu/ml/datasets/KASANDR}} dataset \cite{sidana17}  
contains 15,844,717 interactions of 2,158,859 users in Germany using Kelkoo's\footnote{\url{http://www.kelkoo.fr/}} online advertising platform.
\item {\RecS} is a sample based on historic XING data provided 6,330,562 feedback given by 39,518 users on the job posting items and the items generated by XING's job recommender system.
\item {\PANDOR}\footnote{\url{https://archive.ics.uci.edu/ml/datasets/PANDOR}} is another publicly available dataset for online recommendation \cite{sidana18}  provided by Purch\footnote{\url{http://www.purch.com/}}. The dataset records 2,073,379 clicks generated by 177,366 users of one of the Purch's high-tech website over 9,077 ads they have been shown during one month. 
\end{itemize}
\noindent Table \ref{tab:datasets} presents some detailed statistics about each collection. Among these, we report the average number of clicks (positive feedback) and the average number of items that were viewed but not clicked (negative feedback). As we see from the table,  {\Out}, {\kasandr},  and {\PANDOR} datasets are the

\begin{table}[htbp]
    \centering
    \resizebox{0.95\textwidth}{!}
    {
    \begin{tabular}{lccccc}
    \hline
    Data&$|\mathcal{U}|$&
    $|\mathcal{I}|$ & Sparsity&~~~~Avg. \# of $+$~~~~ & ~Avg. \# of $-$~~~~\\
    \hline
    {\ML}-1M&6,040&3,706&.9553&95.2767& 70.4690\\
   \Out&49,615&105,176&.9997&6.1587& 26.0377 \\
   \PANDOR&177,366&9,077&.9987&1.3266& 10.3632\\
     \NetF&90,137&3,560&.9914&26.1872& 20.2765\\
    \RecS&39,518&28,068&.9943&26.2876&133.9068\\
   \kasandr&2,158,859&291,485&.9999&2.4202& 51.9384\\
        \hline
    \end{tabular}
    }
    \caption{Statistics on the \# of users and items; as well as the sparsity and the average number of $+$ (preferred) and $-$ (non-preferred) items on {\ML}-1M, {\NetF}, {\Out}, {\kasandr} and {\PANDOR} collections after preprocessing.}
    \label{tab:datasets}
\end{table}

\noindent most unbalanced ones in regards to the number of preferred and non-preferred items. To construct the training and the test sets, we discarded users who did not interact over the shown items and sorted all interactions according to time-based on the existing time-stamps related to each dataset. Furthermore, we considered $80\%$ of each user's first interactions (both positive and negative) for training, and the remaining for the test. Table \ref{tab:detail_setting} presents the size of the training and the test sets as well as the percentage of positive feedback (preferred items) for all collections ordered by increasing training size. The percentage of positive feedback is inversely proportional to the size of the training sets, attaining $3\%$ for the largest, \kasandr{} collection. 
\begin{table}[h]
    \centering
    \begin{tabular}{lcccc}
    \hline
    Dataset&$|S_{train}|$~~~~&~~~$|S_{test}|$~~&~~$pos_{train}$~~&~~$pos_{test}$\\
    \hline
    {\ML}-1M&797,758&202,451&58.82&52.39\\
     \Out&1,261,373&336,053&17.64&24.73\\
    \PANDOR&1,579,716&493,663&11.04&12.33\\
    \NetF&3,314,621&873,477&56.27&56.70\\
    \RecS&5,048,653&1,281,909&17.07&13.81\\
    \kasandr&12,509,509&3,335,208&3.36&8.56\\
    \hline
    \end{tabular}
    \caption{Number of interactions used for train and test on each dataset, and the percentage of positive feedback among these interactions.}
    \label{tab:detail_setting}
\end{table}

We also analyzed the distributions of the number of blocks and their size for different collections. Figure \ref{fig:boxplots} (left) shows boxplots representing the logarithm of the number of blocks through their quartiles for all collections. From these plots, it comes out that the distribution of the number of blocks on {\PANDOR}, {\NetF} and {\kasandr} are heavy-tailed with more than the half of the users interacting no more than 10 times with the system. Furthermore, we note that on {\PANDOR} the average number of blocks is much smaller than on the two other collections; and that on all three collections the maximum numbers of blocks are $10$ times more than the average. These plots suggest that a very small number of users (perhaps bots) have an abnormal interaction with the system generating a huge amount of blocks on these three collections. To have a better understanding, Figure \ref{fig:boxplots} (right) depicts the number of blocks concerning their size on {\kasandr}. The distribution of the number of blocks follows a power law distribution and it is the same for the other collections that we did not report for the sake of space. In all collections, the number of blocks having more than $5$ items drops drastically. As in both variants of {\SO} positive and negative items are not sampled for updating the weights, these updates are performed on blocks of small size, and are made very often. 

\begin{figure}[t!]
    \begin{center}
    \begin{tabular}{ccc}
    \begin{tikzpicture}[scale=.6]
        \begin{axis}
            [ 
                ylabel = {$\log_{10}(\text{Number of blocks})$},
                boxplot/draw direction=y,
                xtick={1,2,3,4,5,6},
                xticklabels={\ML-1M, \Out,\PANDOR,\NetF,\kasandr, \RecS},
                x tick label style={font=\footnotesize, rotate=45,anchor=east}
            ]
    \addplot+[mark = *, mark options = {black!30},
        boxplot prepared={
        lower whisker=0,
        lower quartile=0.84,
        median=1.67,
        average = 1.89,
        upper quartile=2.61,
        upper whisker=3.08
    }, color = black!60
    ]coordinates{};
    \addplot+[mark = *, mark options = {black!30},
    boxplot prepared={
        lower whisker=0.301,
        lower quartile=0.47,
        median=0.602,
        average = 0.651,
        upper quartile=1,
        upper whisker=1.46
    }, color = black!60
    ]coordinates{};
       \addplot+[mark = *, mark options = {black!30},
    boxplot prepared={
      lower whisker=0,
      lower quartile=0,
      median=0,
      average = 0.11,
      upper quartile=0.6,
      upper whisker=1.85
    }, color = black!60
    ]coordinates{};
        \addplot+[mark = *, mark options = {black!30},
    boxplot prepared={
      lower whisker=0,
      lower quartile=0.1,
      median=1.07,
      average = 1.34,
      upper quartile=2.12,
      upper whisker=2.77
    }, color = black!60
    ]coordinates{};
        \addplot+[mark = *, mark options = {black!30},
    boxplot prepared={
      lower whisker=0,
      lower quartile=0,
      median=0,
      average = 0.66,
      upper quartile=1.65,
      upper whisker=3.22
    }, color = black!60
    ]coordinates{};
    \addplot+[mark = *, mark options = {black!30},
      boxplot prepared={
      lower whisker=0,
      lower quartile=0.85,
      median= 1.18,
      average = 1.34,
      upper quartile=1.45,
      upper whisker=3.01
    }, color = black!60
    ]coordinates{};
    \end{axis}
\end{tikzpicture}
& ~~~~& 
\begin{tikzpicture}[scale=.6]
\vspace{-8mm}\begin{axis}
[ 
    ybar,
 xmajorgrids, 
 bar width = {2.48em},
 yminorticks=true, 
 ymajorgrids, 
 yminorgrids,
 ylabel={Number of blocks},
 xlabel={Size of the blocks},
 ymin = 0.0,
 ymax= 160000.0,
 symbolic x coords={1-5, 5-10, 10-15, 15-20, 20-25, 25-30, 30-35, 35-40, 40-45},
 xtick=data,
label style={font=\footnotesize},
 ];
 \addplot [fill=gray] coordinates {
      (1-5, 154873)
      (5-10, 61243) 
      (10-15, 24733)
      (15-20, 14330)
      (20-25, 11907)
      (25-30, 7281)
      (30-35, 6120)
    };
\end{axis}
\end{tikzpicture}\\
(a) &~& (b) 
\end{tabular}
        \end{center}
    \caption{(a) Boxplots depicting the logarithm of the number of  blocks through their quartiles for all collections. The median (resp. mean) is represented by the band (resp. diamond) inside the box. The ends of the whiskers represent the minimum and the maximum of the values. (b) Distributions of negative feedback over the blocks in the training set on  {\kasandr}.
    }
    \label{fig:boxplots}
\end{figure}
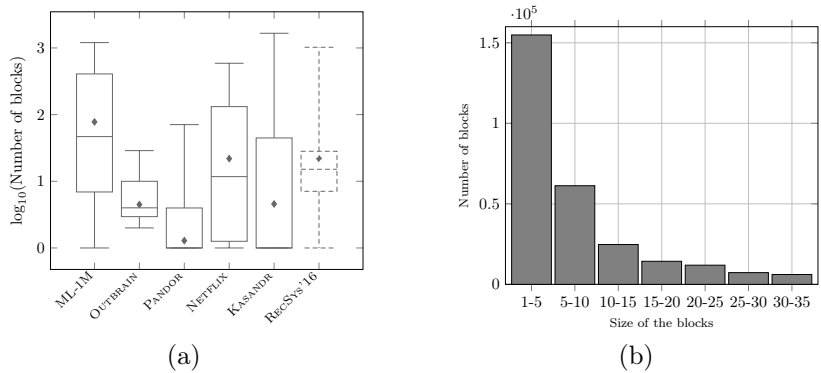
\paragraph{Compared approaches. } To validate the sequential learning approach described in the previous sections, we compared the proposed \SO{} algorithm\footnote{The source code is available at \url{https://github.com/SashaBurashnikova/SAROS}.} with the following approaches.  
\begin{itemize}
    \item {\MostPop} is a non-learning based approach which consists in recommending the same set of popular items to all users.
    \item Matrix Factorization ({\MF}) \cite{Koren08}, is a factor model which decomposes the matrix of user-item interactions into a set of low dimensional vectors in the same latent space, by minimizing a regularized least square error between the actual value of the scores and the dot product over the user and item representations. 
    \item {\BPR} \cite{rendle_09} corresponds to the model described in the problem statement above (Section \ref{sec:LO}), a stochastic  gradient-descent  algorithm, based  on  bootstrap  sampling  of  training  triplets, and {\batch} the batch version of the model which consists in finding  the model parameters $\weight=(\bfU,\bfI)$ by minimizing the ranking loss over all the set of triplets simultaneously (Eq.~\ref{eq:RL}).
   \item {\ProdVec} \cite{GrbovicRDBSBS15},  learns the representation of items using a Neural Networks based model, called word2vec \cite{word_emb}, and performs next-items recommendation using the similarity between the representations of items.
   
    \item {\GRU}$+$ \cite{hidasi2018recurrent} is an extended version of {\GRU} \cite{hidasi} adopted to different loss functions, that applies recurrent neural network with a GRU architecture for session-based recommendation. The approach considers the session as the sequence of clicks of the user and learns  model parameters by optimizing a regularized approximation of the relative rank of the relevant item which favors the preferred items to be ranked at the top of the list.
    
    \item {\caser} \cite{tang2018caser} is a CNN based model that embeds a sequence of clicked items into a temporal image and latent spaces and find local characteristics of the temporal image using convolution filters. 
    
    \item {\SASR} \cite{DBLP:conf/icdm/Kang18} uses  an attention mechanism to capture long-term semantics in the sequence of clicked items  and then predicts the next item to present based on a user's action history.
\end{itemize}

Hyper-parameters of different models and the dimension of the embedded space for the representation of users and items; as well as the regularisation parameter over the norms of the embeddings for all approaches were found by cross-validation. 

We fixed $b$ and $B$, used in $\SO_b$, to respectively the minimum and the average number of blocks found on the training set of each corresponding collection. With the average number of blocks being greater than the median on all collections, the motivation here is to consider the maximum number of blocks by preserving the model from the bias brought by the too many interactions of the very few number of users. For more details regarding the exact values for the parameters, see the Table \ref{parameters}.

\begin{table}[H]
    \centering
    \resizebox{0.95\textwidth}{!}
    {
    \begin{tabular}{c|cccccc}
    \hline
    Parameter&{\ML}&{\Out}&{\PANDOR}&{\NetF}&{\kasandr}&{\RecS}\\
    \hline
    $B$&78&5&2&22&5&22\\
    $b$&1&2&1&1&1&1\\
    Learning rate&.05&.05&.05&.05&.4&.3\\
    \hline
    \end{tabular}
    }
    \caption{Hyperparameter values of $\SO_b$.}
    \label{parameters}
\end{table}

\paragraph{Evaluation setting and results. }
We begin our comparisons by testing {\batch}, \BPR{} and \SO{} approaches over the logistic ranking loss (Eq. \ref{eq:instloss}) which is used to train them. Results on the test, after training the models till the convergence are shown in Table \ref{test_loss}. {\batch} (resp. {\SO}) techniques have the worse (resp. best) test loss on all collections, and the difference between their performance is larger for bigger size datasets. 

\begin{table}[!h]
\centering
{
\begin{tabular}{cc|cccccc}
\hline
Dataset &~~&\multicolumn{6}{c}{Test loss at convergence, Eq.~\eqref{eq:RL}}\\
\cline{3-8}
&~~&\batch&~~&{\BPR}&~~&$\SO_b$&$\SO_m$\\
\hline
\ML-1M &~~&0.744&~~&0.645&~~&\bf{0.608}&0.637\\
\Out &~~& 0.747&~~&0.638&~~&0.635&\bf{0.634}\\
\PANDOR &~~&0.694&~~&0.661&~~&\bf{0.651}&0.666\\
\NetF &~~&0.694&~~&0.651&~~&\bf{0.614}&0.618\\
\kasandr &~~&0.631&~~&0.393&~~&\bf{0.212}&0.257\\
\RecS &~~&0.761&~~&0.644&~~&0.640&\bf{0.616}\\
\hline
\end{tabular}
}
\caption{Comparison between \BPR, \batch{} and \SO{} approaches in terms of test loss at convergence.}
\label{test_loss}
\end{table}

These results suggest that the local ranking between preferred and no-preferred items present in the blocks of the training set reflects better  the preference of users than the ranking of random pairs of items as it is done ine {\BPR} without this sequence information of \textit{viewed but not clicked} and \textit{viewed and clicked} over items in user's sessions. Furthermore, as in {\SO} updates occur after the creation of a block, and that the most of the blocks contain very few items (Figure \ref{fig:boxplots} - right),  weights are updated more often than in {\BPR} or {\batch}. This is depicted in Figure \ref{fig:losses} which shows the evolution of the training error over time for {\batch}, \BPR{} and \SO{} on all collections. As we can see, the training error decreases in all cases and the three approaches converge to the same minimizer of the ranking loss (Eq.~\ref{eq:RL}). This is an empirical evidence of the convergence of $\SO_b$ and $\SO_m$, showing that the sequence of weights found by the algorithm minimizing (Eq. \ref{eq:CLoss}) allows to minimize the general ranking loss (Eq.~\ref{eq:RL}) as it is stated in Theorems 1 and 2.

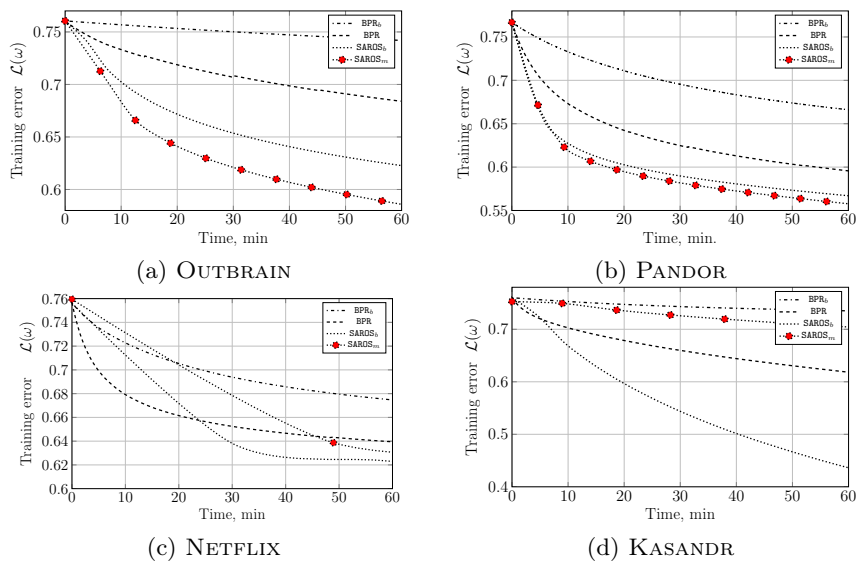
\begin{figure}[h!]
\small
    \centering
    \begin{tabular}{cc}
\begin{tikzpicture}[scale=0.42]
\begin{axis}[ 
 width=\columnwidth, 
 height=0.65\columnwidth, 
 xmajorgrids, 
 yminorticks=true, 
 ymajorgrids, 
 yminorgrids,
 ylabel={Training error ~$\mathcal{L}(\omega)$},
 xlabel = {Time, min},
 ymin = 0.58,
 ymax=0.77,
 xmin = 0,
 xmax = 60,
label style={font=\Large} ,
tick label style={font=\Large}
 ];

  \addplot  [color=black,
                dash pattern=on 1pt off 3pt on 3pt off 3pt,
                mark=none,
                mark options={solid},
                smooth,
                line width=1.2pt]  file {out_batch.txt }; 
 \addlegendentry{ \batch };

  \addplot  [color=black,
                dashed,
                mark=none,
                mark options={solid},
                smooth,
                line width=1.2pt]  file {out_sgd.txt }; 
 \addlegendentry{ \BPR };
  \addplot  [color=black,
                dotted,
                mark=none,
                mark options={solid},
                smooth,
                line width=1.2pt]  file {out_online.txt }; 
 \addlegendentry{ $\SO{}_b$ };
 
 \addplot  [color=black,
                dotted,
                mark=*,
                mark options={scale = 1.5, fill = red},
                smooth,
                line width=1.2pt]  file {out_momentum.txt }; 
 \addlegendentry{ $\SO{}_m$ };
\end{axis}
\end{tikzpicture} & \begin{tikzpicture}[scale=0.42]
\begin{axis}[ 
 width=1.0\columnwidth, 
 height=0.65\columnwidth, 
 xmajorgrids, 
 yminorticks=true, 
 ymajorgrids, 
 yminorgrids,
 ylabel={Training error ~$\mathcal{L}(\omega)$},
 xlabel = {Time, min.},
 ymin = 0.55,
 ymax= 0.78,
 xmin = 0,
 xmax = 60,
label style={font=\Large} ,
tick label style={font=\Large}
 ]
 
 \addplot  [color=black,
                dash pattern=on 1pt off 3pt on 3pt off 3pt,
                mark=none,
                mark options={solid},
                smooth,
                line width=1.2pt]  file {pandor_batch.txt }; 
 \addlegendentry{ \batch }; 
 
   \addplot  [color=black,
                dashed,
                mark=none,
                mark options={solid},
                smooth,
                line width=1.2pt]  file {pandor_sgd.txt }; 
 \addlegendentry{ \BPR };

  \addplot  [color=black,
                dotted,
                mark=none,
                mark options={solid},
                smooth,
                line width=1.2pt]  file {pandor_online.txt }; 
 \addlegendentry{ $\SO{}_b$ };
 
   \addplot  [color=black,
                dotted,
                mark=*,
                mark options={scale = 1.5, fill = red},
                smooth,
                line width=1.2pt]  file {pandor_momentum.txt }; 
 \addlegendentry{ $\SO{}_m$ };
\end{axis}
\end{tikzpicture} \\
(a) \Out & (b) \PANDOR\\
\begin{tikzpicture}[scale=0.4]
\begin{axis}[ 
 width=\columnwidth, 
 height=0.65\columnwidth, 
 xmajorgrids, 
 yminorticks=true, 
 ymajorgrids, 
 yminorgrids,
 ylabel={Training error ~~~~$\mathcal{L}(\omega)$},
 xlabel = {Time, min},
 ymin = 0.6,
 ymax=0.76,
 xmin = 0,
 xmax = 60,
label style={font=\Large} ,
tick label style={font=\Large}
 ];
 
  \addplot  [color=black,
                dash pattern=on 1pt off 3pt on 3pt off 3pt,
                mark=none,
                mark options={solid},
                smooth,
                line width=1.2pt]  file {netflix_batch.txt }; 
 \addlegendentry{ \batch };

  \addplot  [color=black,
                dashed,
                mark=none,
                mark options={solid},
                smooth,
                line width=1.2pt]  file {netflix_sgd.txt }; 
 \addlegendentry{ \BPR };
  \addplot  [color=black,
                dotted,
                mark=none,
                mark options={solid},
                smooth,
                line width=1.2pt]  file {netflix_online.txt }; 
 \addlegendentry{ $\SO{}_b$ };
 \addplot  [color=black,
                dotted,
                mark=*,
                mark options={scale = 1.5, fill = red},
                smooth,
                line width=1.2pt]  file {netflix_momentum.txt }; 
 \addlegendentry{ $\SO{}_m$ };
 \end{axis}
\end{tikzpicture} & \begin{tikzpicture}[scale=0.42]
\begin{axis}[ 
 width=\columnwidth, 
 height=0.65\columnwidth, 
 xmajorgrids, 
 yminorticks=true, 
 ymajorgrids, 
 yminorgrids,
 ylabel={Training error ~$\mathcal{L}(\omega)$},
 xlabel = {Time, min},
 ymin = 0.4,
 ymax=0.78,
 xmin = 0,
 xmax = 60,
label style={font=\Large} ,
tick label style={font=\Large}
 ];

 \addplot  [color=black,
                dash pattern=on 1pt off 3pt on 3pt off 3pt,
                mark=none,
                mark options={solid},
                smooth,
                line width=1.2pt]  file {kassandr_batch.txt }; 
 \addlegendentry{ \batch };

  \addplot  [color=black,
                dashed,
                mark=none,
                mark options={solid},
                smooth,
                line width=1.2pt]  file {kassandr_sgd.txt }; 
 \addlegendentry{ \BPR };
 
  \addplot  [color=black,
                dotted,
                mark=none,
                mark options={solid},
                smooth,
                line width=1.2pt]  file {kassandr_online.txt }; 
 \addlegendentry{$\SO{}_b$};
    \addplot  [color=black,
                dotted,
                mark=*,
                mark options={scale = 1.5, fill = red},
                smooth,
                line width=1.2pt]  file {kassandr_momentum.txt }; 
 \addlegendentry{ $\SO{}_m$ };

\end{axis}
\end{tikzpicture}\\
(c) \NetF & (d) \kasandr\\
\end{tabular}
\caption{Evolution of the loss on training sets for both {\batch}, {\BPR} and {\SO} as a function of time in minutes for all collections.}
    \label{fig:losses}
\end{figure}

We also compare the performance of all the approaches on the basis of the common ranking metrics, which are the Mean Average Precision at rank $K$ (\mapk) over all users  defined as $\mapk=\frac{1}{N}\sum_{u=1}^{N}\apk(u)$, where $\apk(u)$ is the average precision of  preferred items of user $u$ in the top $K$ ranked ones; and the Normalized Discounted Cumulative Gain at rank $K$ (\ndcgk) that computes the ratio of the obtained ranking to the ideal case and allow to consider not only binary relevance as in Mean Average Precision, $\ndcgk = \frac{1}{N}\sum_{u=1}^{N}\frac{\dcgk(u)}{\idcgk(u)}$, where $\dcgk(u) = \sum_{i=1}^{K}\frac{2^{rel_{i}}-1}{\log_{2}(1+i)}$, $rel_{i}$ is the graded relevance of the item at position $i$; and $\idcgk(u)$ is $\dcgk(u)$ with an ideal ordering equals to $\sum_{i=1}^{K}\frac{1}{\log_{2}(1+i)}$ for $rel_{i}\in[0,1]$  \cite{Manning:2008}.\\

To estimate the importance of the maximum number of blocks ($B$) for $\SO_b$, we explore the dependency between quality metrics {\mapk} and {\ndcgk} on {\ML-1M} and {\PANDOR} collections (Figure \ref{fig:map5_blocks}). The latter records the clicks generated by users on one of Purch’s high-tech website and it was subject to bot attacks \cite{sidana18}. For this collection, large values of $B$ affects  {\mapk} while the measure reaches a plateau on {\ML-1M}. The choice of this hyperparameter may then have an impact on the results. As future work, we are investigating the modelling of bot attacks by studying the effect of long memory in the blocks of no-preferred and preferred items in small and large sessions with the aim of automatically fixing this threshold $B$.

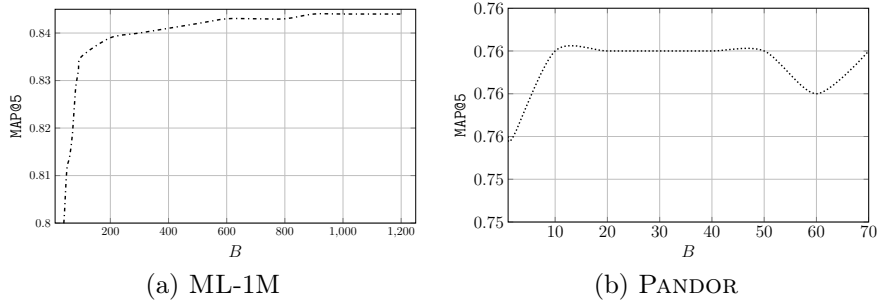
\begin{figure}[h!]
    \centering
    \begin{tabular}{cc}



\begin{tikzpicture}[scale=0.45]
\begin{axis}[ 
 width=1.0\columnwidth, 
 height=0.65\columnwidth, 
 xmajorgrids, 
 yminorticks=true, 
 ymajorgrids, 
 yminorgrids,
 ylabel={\mapfive},
 xlabel = {$B$},
 ymin = 0.800,
 ymax= 0.845,
 xmin = 10,
 xmax = 1250,
label style={font=\Large} ,
 ];

 \addplot  [color=black,
                dash pattern=on 1pt off 3pt on 3pt off 3pt,
                mark=none,
                mark options={solid},
                smooth,
                line width=1.2pt]  file {ml_blocks_tuning.txt};

\end{axis}
\end{tikzpicture}
&
\begin{tikzpicture}[scale=0.45]
\begin{axis}[ 
 width=\columnwidth, 
 height=0.65\columnwidth, 
 xmajorgrids, 
 yminorticks=true, 
 ymajorgrids, 
 yminorgrids,
 ylabel={{\mapfive}},
 xlabel = {$B$},
 ymin = 0.754,
 ymax=0.759,
 xmin = 1,
 xmax = 70,
legend style={font=\Large},
label style={font=\Large} ,
tick label style={font=\Large}
 ];

  \addplot  [color=black,
                dotted,
                mark=none,
                mark options={solid},
                smooth,
                line width=1.2pt]  file {pandor_map5_blocks}; 
\end{axis}
\end{tikzpicture}\\
(a) \ML-1M & (b) \PANDOR\\
\end{tabular}
\caption{Evolution of {\mapfive} with respect to largest number of allowed blocks, $B$.}
    \label{fig:map5_blocks}
\end{figure}
Table \ref{tab:online_vs_all_ndcg_1h} presents {\mapfive} and {\mapten} (top), and {\ndcgfive} and {\ndcgten} (down)  of  all approaches over the test sets of the different collections.
\begin{table*}[h!]
    \centering
     \resizebox{1.0\textwidth}{!}{\begin{tabular}{c|cccccc|cccccc}
    \hline
      &\multicolumn{6}{c|}{\mapfive}&\multicolumn{6}{c}{\mapten}\\
     \cline{2-13}
     &{\ML-1M}&\Out&\PANDOR&{\NetF}&{\kasandr}&{\RecS}&{\ML-1M}&\Out&\PANDOR&{\NetF}&{\kasandr}&{\RecS}\\
     \hline
     \MostPop  & .074&.007 &.003 &.039 & .002&.003& .083&.009 &.004 &.051 &.3e-5&.004 \\
     \ProdVec  & .793& .228& .063& .669& .012&.210& .772& .228&.063 &.690 &.012&.220 \\
     \MF       & .733&.531 &.266 & .793& .170&.312& .718& .522& .267& .778& .176&.306\\
    \batch     & .713& .477& .685& .764& .473&.343& .688& .477& .690& .748& .488&.356\\
     \BPR      & \underline{.826}& \underline{.573}& \underline{.734}& \underline{.855}& .507&\bf{.578}&\underline{.797}& \underline{.563}& \bf{.760} & \underline{.835}& .521&\bf{.571}\\
     {\GRU}$+$ & .777& .513 & .673& .774& \underline{.719}& .521& .750& .509&.677 &.757 &\underline{.720}&.500\\
      {\caser} & .718& .471& .522& .749& .186&.218& .694& .473& .527& .733& .197&.218\\
     {\SASR}   & .776 & .542& .682& .819& .480&.521& .751& .534& .687& .799& .495&.511\\
     $\SO_m$&.816&.577&.720&.857&.644&.495&.787&.567&.723&.837&.651&.494\\
     $\SO_b$     & \bf{.832}& \bf{.619}& \bf{.756}& \bf{.866}& \bf{.732}&\underline{.570}& \bf{.808}& \bf{.607}& \underline{.759}& \bf{.846}& \bf{.747}&\underline{.561}\\

     \hline
    \end{tabular}
    }
~\\
~\\
     \resizebox{1.0\textwidth}{!}{\begin{tabular}{c|cccccc|cccccc}
    \hline
     &\multicolumn{6}{c|}{\ndcgfive}&\multicolumn{6}{c}{\ndcgten}\\
     \cline{2-13}
     &{\ML-1M}&\Out&\PANDOR&{\NetF}&{\kasandr}&{\RecS}&{\ML-1M}&\Out&\PANDOR&{\NetF}&{\kasandr}&{\RecS}\\
     \hline
     \MostPop  & .090&.011 &.005 &.056 & .002&.004& .130&.014 &.008 &.096 &.002&.007\\
     \ProdVec  & .758& .232& .078& .712& .012&.219& .842& .232&.080 &.770 &.012&.307 \\
     \MF       & .684&.612 &.300 & .795& .197&.317& .805& .684& .303& .834& .219&.396\\
    \batch     & .652& .583& .874& .770& .567&.353& .784& .658& .890& .849& .616&.468\\
     \BPR      & \underline{.776}& \underline{.671}& \underline{.889}& \underline{.854}& .603 &\bf{.575}&\underline{.863}& \underline{.724}& \underline{.905}& \underline{.903}& .650&\bf{.673}\\
     {\GRU}$+$ & .721& .633 & .843& .777& \underline{.760}&.507& .833& .680&.862 &.854 &\underline{.782}&.613\\
      {\caser} & .665& .585& .647& .750& .241&.225& .787& .658& .666& .834& .276&.225\\
     {\SASR}   & .721 & .645& .852& .819& .569&.509& .832& .704& .873& .883& .625&.605\\
     {$\SO_m$}&.763&.674&.885&.857&.735&.492&.858&.726&.899&.909&.765&.603\\
     {$\SO_b$}     & \bf{.788}& \bf{.710}& \bf{.904}& \bf{.866}& \bf{.791}&\underline{.563}& \bf{.874}& \bf{.755}& \bf{.917}& \bf{.914}& \bf{.815}&\underline{.662}\\

     \hline
    \end{tabular}
    }
    \caption{Comparison between \MostPop, \ProdVec, \MF, \batch, \BPR{}, {\GRU}$+$, \SASR, \caser, and \SO{} approaches in terms of \mapfive{} and \mapten (top), and \ndcgfive{} and \ndcgten (down). Best performance is in bold and the second best is underlined.}
    \label{tab:online_vs_all_ndcg_1h}
\end{table*}
The non-machine learning method, {\MostPop},  gives results of an order of magnitude lower than the learning based approaches. Moreover,  the factorization model {\MF} which predicts clicks by matrix completion is less effective when dealing with implicit feedback than ranking based models especially on large datasets where there are fewer interactions. We also found that embeddings found by ranking based models, in the way that the user preference over the pairs of items is preserved in the embedded space by the dot product, are more robust than the ones found by {\ProdVec}. When comparing {\GRU}$+$ with {\BPR} that also minimizes the same surrogate ranking loss, the former outperforms it in case of {\kasandr} with a huge imbalance between positive and negative interactions.

This is mainly because  {\GRU}$+$ optimizes an approximation of the relative rank that favors interacted items to be in the top of the ranked list while the logistic ranking loss, which is mostly related to the Area under the ROC curve \cite{Usunier:1121}, pushes up clicked items for having good ranks in average. However, the minimization of the logistic ranking loss over blocks of very small size pushes the clicked item to be ranked higher than the no-clicked ones in several lists of small size and it has the effect of favoring the clicked item to be at the top of the whole merged lists of items. Moreover,  it comes out that {\SO} is the most competitive approach, performing better than other approaches almost over all collections, except on {\RecS} and {\PANDOR} (for {\mapten}), where \BPR{} gave the best results. We suspect that this is due to a lack of long dependencies in the sequence of items in the formed blocks which cannot be captured by {\SO}. We plan to explore this issue as part of our future work.

\section{Conclusion}\label{sec:Conclusion}

The contributions of this paper are twofold. First, we proposed {\SO},  a novel learning framework for large-scale Recommender Systems that sequentially updates the weights of a ranking function user by user over blocks of items ordered by time where each block is a sequence of negative items followed by a last positive one. The main hypothesis of the approach is that the preferred and no-preferred items within a local sequence of user interactions express better the user preference than when considering the whole set of preferred and no-preferred items independently one from another. We presented two variants of the approach; in the first model parameters are updated user per user over blocks of items constituted by a sequence of unclicked items followed by a clicked one. The parameter updates are discarded for users who interact very little or a lot with the system. The second variant, is based on the momentum technique as a means of damping oscillations.  The second contribution is a theoretical analysis of the proposed approach which bounds the deviation of the ranking loss concerning the sequence of weights found by both variants of the algorithm and its minimum in the general case of non-convex ranking loss.  Empirical results conducted on six real-life implicit feedback datasets support our founding and show that the proposed approach is significantly faster than the common batch and online optimization strategies that consist in updating the parameters over the whole set of users at each epoch, or after sampling random pairs of preferred and no-preferred items. The approach is also shown to be highly competitive concerning state of the art approaches on \texttt{MAP} and \texttt{NDCG} measures.

\section{Acknowledgements}
This work at Los Alamos was supported by the U.S. Department of Energy through the Los Alamos National Laboratory as part of LANL LDRD 20210078DR and 20190351ER. Los Alamos National Laboratory is operated by Triad National Security, LLC, for the National Nuclear Security Administration of U.S. Department of Energy (Contract No. 89233218CNA000001).

\vskip 0.2in
\bibliography{references.bib}
\bibliographystyle{chicago}

\end{document}